\documentclass[12pt]{article}

\usepackage{amsfonts,amsmath,amsthm}
\usepackage{enumerate}
\usepackage{a4wide}

\newcommand{\ttz}{\mathtt{z}}
\newcommand{\sur}{\mathbin{/}}
\newcommand{\twotwo}[4]
{ \begin{bmatrix} #1 & #2 \\ #3 & #4 \end{bmatrix}  }
\newcommand{\poly}{\mathop{\mathrm{poly}}}

\newcommand{\floor}[1]{\left\lfloor #1 \right\rfloor}

\newcommand{\N}{\mathbb{N}}
\newcommand{\Nast}{\mathbb{N}^\ast}
\newcommand{\Q}{\mathbb{Q}} 
\newcommand{\C}{\mathbb{C}} 
 
\newcommand{\defeq}{\mathrel{\mathop:}=}

\newtheorem{definition}{Definition}
\newtheorem{proposition}{Proposition}
\newtheorem{lemma}{Lemma}
\newtheorem{property}{Property}
\newtheorem{corollary}{Corollary}

\begin{document}

\title{A simple, polynomial-time algorithm for the matrix torsion problem}
\author{Fran\c{c}ois Nicolas} 
\maketitle

\begin{abstract}
The Matrix Torsion Problem (MTP) is:  
given a square matrix $M$ with rational entries, decide whether
two distinct powers of $M$ are equal.
It has been shown by Cassaigne and the author that the MTP  reduces to the Matrix Power Problem (MPP) in polynomial time \cite{CassaigneN08}: 
given two square matrices $A$ and $B$ with rational entries, the  MTP is to decide whether $B$ is a power of $A$.
Since the MPP is decidable in polynomial time \cite{KannanL86}, 
it is also the case of the MTP.
However, the algorithm for  MPP is highly non-trivial.
The aim of this note is to present a simple, direct,  polynomial-time algorithm for the MTP.
\end{abstract} 

\section{Introduction}

As usual  $\N$, $\Q$ and $\C$ denote 
the semiring of non-negative integers, 
the field of rational numbers and the field of complex numbers, 
respectively.

\begin{definition}[Torsion] 
Let $M$ be a  square matrix over $\C$.
We say that $M$ is \emph{torsion} if it satisfies the following three equivalent assertions.
\begin{enumerate}[$(i)$.]
\item There exist   $p$, $q \in \N$ such that $p \ne q$ and  $M^p = M^q$.
\item The multiplicative  semigroup $\left\{ M, M^2, M^3, M^4, \dotsc \right\}$ has  finite cardinality.
\item The sequence $(M, M^2, M^3, M^4, \dotsc)$ is eventually periodic.
\end{enumerate}
\end{definition}

The aim of this note is to present a polynomial-time algorithm for the following decision problem:

\begin{definition}
The \emph{Matrix Torsion Problem (MTP)} is:
given as input a square matrix $M$ over $\Q$, decide whether $M$ is torsion.
\end{definition} 

For every square matrix $M$ over $\Q$, the size of $M$ is defined as the order of $M$ plus the sum of the lengths of the binary encodings over all entries of $M$.

\paragraph{Previous work.}
The Matrix Power Problem (MPP)  is:  given two square matrices $A$ and $B$ over $\Q$ decide whether there exists $n \in \N$ such that  such that $A^n = B$.
Kannan and Lipton showed that the MPP is decidable in polynomial time \cite{KannanL86}.
It is rather easy to prove that the MTP is decidable in polynomial time by reduction to  the MPP (see Section~\ref{sec:comment}).
However, the original algorithm for the MPP is highly non-trivial.
Hence, a simple, direct algorithm is still interesting.

\section{Generalities}

Throughout this paper, $\ttz$ denotes an indeterminate.
The next result is well-known and plays a crucial role in the paper. 
  
\begin{proposition}[Multiple root elimination] \label{prop:mult-root}
For every polynomial $\nu(\ttz)$ over $\C$, 
the polynomial 
$$
\pi(\ttz) \defeq \dfrac{\nu(\ttz)}{\gcd(\nu'(\ttz), \nu(\ttz))} 
$$
satisfies the following two properties:
\begin{itemize} 
\item $\nu(\ttz)$ and $\pi(\ttz)$ have the same complex roots, and
\item $\pi(\ttz)$ has no multiple roots.
\end{itemize} 
\end{proposition}

The set of all positive integers is denoted $\Nast$. 

For every  $n \in \Nast$, the  \emph{$n^\text{th}$ cyclotomic polynomial}  is denoted $\gamma_n(\ttz)$:
$\gamma_n(\ttz) = \prod (\ttz - u)$ where the product is over all primitive $n^\text{th}$ roots of unity $u \in \C$. 
It is well-known that $\gamma_n(\ttz)$ is a monic integer polynomial and that the following three properties hold \cite{LangAlgebra}:

\begin{property}  \label{prop:zm-1}
For every $m \in \Nast$, $\ttz^m - 1 = \prod_{n \in D_m} \gamma_n(\ttz)$, where $D_m$ denotes the set of  all positive divisors of $m$.
\end{property} 

\begin{property} \label{prop:irreductible}
For every $n\in \Nast$, $\gamma_n(\ttz)$ is irreducible over $\Q$.
\end{property}

\emph{Euler's totient function} is denoted $\phi$: for every $n \in \Nast$, $\phi(n)$ equals the number of $k \in  \{ 1, 2, \dotsc, n \}$ such that $\gcd(k, n) = 1$.

\begin{property} \label{prop:degre-phi}
For every $n \in \Nast$, $\gamma_n(\ttz)$ is of degree $\phi(n)$.
\end{property} 

The next lower bound for Euler's totient function is far from optimal.
However,
 it is sufficient for our purpose.
 
\begin{proposition} \label{prop:totient-ineq}
For every $n\in \Nast$,
$\phi(n)$ is greater than or equal to $\sqrt{n \sur 2}$.
\end{proposition} 

A proof of Proposition~\ref{prop:totient-ineq} can be found in appendix.
Noteworthy is that $\dfrac{\phi(n) \ln \ln n}{n}$ converges to a positive, finite limit as $n$ tends to $\infty$. \cite{HardyWright}.

\section{The new algorithm}

\begin{definition} 
For every $n\in \Nast$, let 
$\displaystyle 
\pi_n(\ttz) \defeq \prod_{j = 1}^n \gamma_j(\ttz)
$.
\end{definition}

To prove that the MTP is decidable in polynomial time, we  prove that 
\begin{itemize}
\item  a  $d$-by-$d$ matrix over $\Q$ is torsion if and only if it  annihilates $\ttz^d \pi_{2 d^2}(\ttz)$, and that 
\item $\ttz^d \pi_{2 d^2}(\ttz)$ is computable from $d$ in $\poly(d)$ time.
\end{itemize}

\begin{proposition} \label{th:critere-torsion}
Let $d$, $n \in \Nast$ be such that for every integer $m$ greater than $n$, $\phi(m)$ is greater than $d$.
For every  $d$-by-$d$ matrix $M$ over $\Q$ the following three assertions are equivalent.
\begin{enumerate}[$(i)$.]
\item $M$ is torsion.
\item $M$ annihilates  $\ttz^d \pi_n(\ttz)$.
\item $M$ satisfies $M^{n! +d } = M^d$.
\end{enumerate}
\end{proposition}

\begin{proof}
The implication  $(iii)  \implies (i)$ is clear.
Moreover, it follows from Property~\ref{prop:zm-1} that $\pi_n(\ttz)$ divides $\ttz^{n!} - 1$, and thus 
$\ttz^d \pi_n(\ttz)$ divides $\ttz^{n! + d } - \ttz^d$.
Therefore, $(ii) \implies (iii)$ holds.
Let us now show $(i) \implies (ii)$. 

Assume that $M$ is torsion.
Let $p$, $q \in \N$ be such that $p < q$ and $M^p = M^q$.
Let $\mu(\ttz)$ denote the minimal polynomial of $M$ over $\Q$.
Since $M^q - M^p$ is a zero matrix,
$\mu(\ttz)$  divides 
$\ttz^q - \ttz^p$.
By Property~\ref{prop:zm-1}, 
$\ttz^q - \ttz^p$ can be factorized in the form 
$\ttz^q - \ttz^p =  \ttz^{p} \prod_{j \in D_{q - p}} \gamma_j(\ttz)$;
by Property~\ref{prop:irreductible} all factors are irreducible  over $\Q$.
Hence, $\mu(\ttz)$ can be written in the form 
$\mu(\ttz) = \ttz^k \prod_{j \in J} \gamma_j(\ttz)
$
for some integer $k$ satisfying  $0 \le k \le p$  and some $J \subseteq D_{q - p}$.
Besides, the Cayley-Hamilton theorem implies that  $\mu(\ttz)$ divides the characteristic polynomial of $M$ which is of degree $d$.
Therefore, $d$ is not smaller than the degree of $\mu(\ttz)$.
Since the degree of $\mu(\ttz)$ equals $k + \sum_{j \in J} \phi(j)$ by Property~\ref{prop:degre-phi}, 
we have $k \le d$ and $\max J  \le n$.
Hence,  $\mu(\ttz)$ divides $\ttz^d \pi_n(\ttz)$.
\end{proof}

Combining Propositions~\ref{prop:totient-ineq} and~\ref{th:critere-torsion}, we obtain that for every $d \in \Nast$, a  $d$-by-$d$ matrix over $\Q$ is torsion if and only if it annihilates the polynomial  $\ttz^d \pi_{2 d^2}(\ttz)$.
To conclude the paper, 
it remains to explain how to compute $\pi_{n}(\ttz)$ in $\poly(n)$ time  from any  $n \in \Nast$ taken as input.
The idea is to rely on Proposition~\ref{prop:mult-root}.

\begin{definition} 
For every  $n \in \Nast$, let 
$\displaystyle 
\nu_n(\ttz) \defeq \prod_{j = 1}^n (\ttz^j - 1) 
$.
\end{definition} 

Let $n \in \Nast$.
Clearly, $\nu_n(\ttz)$ is computable in $\poly(n)$ time.
Moreover, it follows from Property~\ref{prop:zm-1} that
$$
\nu_n(\ttz) = \prod_{j = 1}^n \left( \gamma_j(\ttz) \right)^{\floor{ n \sur j }} \, ,
$$
and thus $\nu_n(\ttz)$ and $\pi_n(\ttz)$ have the same roots.
Since  $\pi_n(\ttz)$ has no multiple roots, 
Proposition~\ref{prop:mult-root} yields a way to compute $\pi_n(\ttz)$ from $\nu_n(\ttz)$ in polynomial time:

\begin{proposition} \label{prop:pi-n-nu-n}
For every $n \in \Nast$, $\pi_n(\ttz) \defeq \dfrac{\nu_n(\ttz)}{\gcd(\nu'_n(\ttz), \nu_n(\ttz))}$.
\end{proposition}

\section{Comments}  \label{sec:comment}

\paragraph{The failure of the naive approach.}

Combining Propositions~\ref{prop:totient-ineq} and~\ref{th:critere-torsion}, we obtain:

\begin{corollary}[Mandel and Simon {\cite[Lemma~4.1]{MandelS77}}] \label{cor:mandel}
Let $d \in \Nast$.
Every $d$-by-$d$ torsion matrix $M$ over $\Q$ satisfies $M^{(2d^2)! + d} = M^d$. 
\end{corollary} 

It follows from Proposition~\ref{cor:mandel} that the MTP is decidable.
However, such an approach does not yield a polynomial-time algorithm for the MTP:

\begin{proposition} \label{prop:exponential}
Let $t : \Nast \to \Nast$ be a  function such that for each $d \in \Nast$, every $d$-by-$d$ torsion  matrix  $M$ over $\Q$ satisfies $M^{t(d) + d} = M^d$.
Then, $t$ has exponential growth.
\end{proposition} 

\begin{proof}
For every $n \in \Nast$, let $\ell(n)$ denote the least common multiple of all positive integers less than or equal to $n$: $\ell(3) = 6$, $\ell(4) = 12$, $\ell(5) = \ell(6) = 60$, \emph{etc}.
It is well-known that $\ell$ has exponential growth:
for every $n \in \Nast$,  $\ell(2n) \ge \binom{2n}{n} \ge 2^n$  \cite{HardyWright}.

For every $d$-by-$d$ non-singular  matrix $M$ over $\Q$, $M^{t(d)}$ is the identity matrix.
Besides, for every  integer $k$ with $1 \le k \le d$, there exists a  $d$-by-$d$ permutation matrix that generates a cyclic group of order $k$, and thus $k$ divides $t(d)$.
It follows that $\ell(d)$ divides $t(d)$.
\end{proof}

\paragraph{Reducing the MTP to the MPP.} 
For the sake of completeness, let us describe the reduction from the MTP to the MPP.
Let $d \in \Nast$ and let $M$ be a $d$-by-$d$ matrix over $\Q$. 
Let 
$N_2  \defeq \twotwo{0}{1}{0}{0}$, 
$A  \defeq \twotwo{M^d}{O}{O}{N_2}$,
$O_2  \defeq \twotwo{0}{0}{0}{0}$ 
and 
$B \defeq \twotwo{M^d}{O}{O}{O_2}$
where $O$ denotes  both the $d$-by-two zero matrix and its transpose. 
It is clear that $A$ and $B$ are two $(d + 2)$-by-$(d + 2)$ matrices over $\Q$.
Moreover, there exists $n \in \N$ such that $A^n = B$ if and only if $M$ is torsion
 \cite{CassaigneN08}.

\section{Open question}

Let $d \in \Nast$.
It follows from Corollary~\ref{cor:mandel} that for every $d$-by-$d$ torsion matrix $M$ over $\Q$, 
 the sequence $(M^d, M^{d + 1}, M^{d + 2}, M^{d + 3}, \dotsc)$ is (purely) periodic with period at most $(2d^2)!$.
 Hence, the maximum cardinality of $\{ M^d, M^{d + 1}, M^{d + 2}, M^{d + 3}, \dotsc \}$, over all $d$-by-$d$ torsion matrices $M$ over $\Q$, is well-defined.
To our knowledge, its asymptotic  behavior as $d$ goes to infinity is unknown.

\bibliography{bibmat}
\bibliographystyle{plain}


\section*{Proof of Proposition~\ref{prop:totient-ineq}}

The following two properties of Euler's totient function are well-known.

\begin{property} \label{prop:ppow}
For every prime number $p$ and every  $v \in \Nast$,
$\phi(p^v) =  p^{v - 1} (p - 1)$. 
\end{property} 

\begin{proof}
For every integer $k$, $\gcd(k, p^v)$ is distinct from one if and only if $p$ divides $k$.
From that we deduce the equality 
\begin{equation}  \label{eq:pr-pr}
\left\{ k \in \{ 1, 2, \dotsc, p^v \} : \gcd(k, p^v) \ne 1 \right\} 
= \left\{ p q : q \in \{ 1, 2, \dotsc, p^{v - 1} \} \right\} \, .
\end{equation} 
Besides, it is easy to see that the left-hand side of Equation~\eqref{eq:pr-pr} has cardinality $p^v - \phi(p^v)$ while its right-hand side has cardinality  $p^{v - 1}$. 
\end{proof} 

\begin{property}\label{prop:mult} 
For every $m$, $n \in \Nast$,  
$\phi(mn) = \phi(m) \phi(n)$ whenever  $\gcd(m, n) = 1$.
\end{property} 

Property~\ref{prop:mult} is consequence of the Chinese remainder theorem.
It states that $\phi$ is \emph{multiplicative}.

\begin{lemma} \label{lem:golden}
For every real number $x \ge 3$, 
$\sqrt{x}$ is less than $x - 1$.
\end{lemma}

\begin{proof}
The two roots of the quadratic polynomial $\ttz^2 - \ttz - 1$ are $\frac{1 + \sqrt{5}}{2} \approx 1.618$ and $\frac{1 - \sqrt{5}}{2} \approx -0.618$;
 they  are  both smaller than $\sqrt{3} \approx 1.732$.
Therefore, $y^2 - y - 1$ is positive for every real number $y  \ge \sqrt{3}$.
Since for every real number $x \ge 3$, $\sqrt{x}$ is not less than $\sqrt{3}$, $(x - 1) - \sqrt{x} = \left( \sqrt{x} \right)^2 - \sqrt{x} - 1$ is positive.
\end{proof}

\begin{lemma} \label{claim:ineg-phi}
Let $p$ and $v$ be two integers with $p \ge 2$ and $v \ge 1$.
Inequality $p^{v / 2} \le p^{v - 1} (p - 1) $ holds if and only if $(p, v) \ne (2, 1)$.
\end{lemma}

\begin{proof}
If $(p, v) = (2, 1)$ then  $p^{v \sur  2} = \sqrt{2}$ is greater than $1 = p^{v - 1} (p - 1)$.
If $v \ge 2$ then  $v \sur  2 \le v - 1$, 
and thus  $p^{v \sur  2} \le  p^{v - 1} \le p^{v - 1} (p - 1)$ follows.
If $v = 1$ and $p \ge 3$ then $p^{v \sur  2} = \sqrt{p}$ is less than $p - 1 = p^{v - 1} (p - 1)$ according to Lemma~\ref{lem:golden}.
\end{proof}

\begin{lemma} \label{lem:phi-sqrt}
Let $n \in \Nast$.
If $n$ is odd or if four divides $n$ then  $\phi(n)$ is greater than or equal to $\sqrt{n}$.
\end{lemma}

\begin{proof}
It is clear that $\phi(1) = 1 = \sqrt{1}$.
Let $n$ be an integer greater than one.
Write $n$ in the form 
$$
n = \prod_{i = 1}^r p_i^{v_i}
$$ 
where $r$, $v_1$, $v_2$, \ldots, $v_r$ are positive integers and where $p_1$, $p_2$, \ldots, $p_r$ are pairwise distinct prime numbers.
Properties~\ref{prop:ppow} and~\ref{prop:mult} yield:
$$
\phi(n) =  \prod_{i = 1}^r \phi(p_i^{v_i}) =  \prod_{i = 1}^r p_i^{v_i - 1} (p_i - 1) \, .
$$ 
Assume  either that $n$ is odd or that four divides $n$.
Then, for each index $i$ with  $1 \le i \le r$, $(p_i, v_i)$ is distinct from $(2, 1)$ and thus inequality $p_i^{v_i \sur  2} \le p_i^{v_i - 1} (p_i - 1)$ holds. 
From that we deduce 
$$
\phi(n) \ge \prod_{i = 1}^r  p_i^{v_i \sur  2}  = \sqrt{n} \, .
$$
\end{proof} 

\begin{proof}[Proof of Proposition~\ref{prop:totient-ineq}]
If $n$ is odd or if four divides $n$ then Lemma~\ref{lem:phi-sqrt} ensures $\phi(n) \ge \sqrt{n} \ge  \sqrt{n / 2}$.
Conversely,
 assume  that $n$ is even and that four does not divide $n$: 
 there exists an odd integer $n'$ such that $n = 2 n'$. 
We have 
\begin{itemize}
 \item $\phi(n) = \phi(2) \phi(n') =  \phi(n')$ according to Property~\ref{prop:mult}, 
 and 
\item $\phi(n') \ge \sqrt{n'} = \sqrt{n / 2}$ by Lemma~\ref{lem:phi-sqrt}. 
\end{itemize}
From that we deduce $\phi(n) \ge  \sqrt{n / 2}$.
\end{proof}

\end{document}